\documentclass[letterpaper,12pt]{article}
\usepackage{tabularx} 
\usepackage{amsmath}  
\usepackage{graphicx} 
\usepackage[margin=1in,letterpaper]{geometry} 
\usepackage{cite} 
\usepackage{authblk}
\usepackage{amssymb}
\usepackage{amsthm}
\providecommand{\keywords}[1]{\textbf{Keywords--} #1}
\newtheorem{theorem}{Theorem}
\newtheorem{definition}{Definition}
\newtheorem{lemma}{Lemma}

\DeclareRobustCommand{\rchi}{{\mathpalette\irchi\relax}}
\newcommand{\irchi}[2]{\raisebox{\depth}{$#1\chi$}} 

\title{Entropy and Compression: A simple proof of an inequality of Khinchin-Ornstein-Shields}
\author[1]{R. Aragona\thanks{
{riccardo.aragona@univaq.it} - phone number: +390862434723 (Riccardo Aragona)}}
\author[1]{F. Marzi\thanks{
{francesca.marzi@univaq.it} (Francesca Marzi)}}
\author[1,2]{F. Mignosi\thanks{
{filippo.mignosi@univaq.it} (Filippo Mignosi)}}
\author[3]{M. Spezialetti\thanks{
{matteo.spezialetti@guest.univaq.it} (Matteo Spezialetti)\protect\vspace{2mm}\\
\indent \,\,\, R. Aragona is member of INdAM-GNSAGA (Italy).}}
\affil[1]{DISIM, University of L'Aquila\protect\\
Via Vetoio, I-67100 Coppito (L'Aquila), Italy
}
\affil[2]{ICAR-CNR\protect\\
Via Ugo la Malfa 153, 90146 Palermo, Italy
}
\affil[3]{University of Naples ``Federico II''\protect\\
Via Claudio 21, 80125 Napoli, Italy
}

\date{}

\begin{document}

\maketitle

\abstract{This paper concerns the folklore statement that ``entropy is a lower bound for compression''. More precisely we derive from the entropy theorem a simple proof of a pointwise inequality firstly stated by Ornstein and Shields and which is the almost-sure  version  of an average inequality firstly stated by Khinchin in 1953. We further give an elementary  proof of original Khinchin inequality that can be used as an exercise for Information Theory students and we conclude by giving  historical and technical notes of such inequality.}

\vspace{2mm}
\noindent

\vspace{2mm}

\hfill{\small{\bf{In memoriam of Professor Aldo de Luca}}}

\section{Introduction and notation}\label{sec:intro}This paper concerns the folklore statement that ``entropy is a lower bound for compression''. Whilst almost every expert of the field knows this statement, both the history and the most general mathematical formulation of it, that is an inequality due to Ornstein and Shields in 1990 \cite{OrnsteinShields90}, are much less known. 
The main objective of this note is to give a simple proof of the 
Ornstein-Shields inequality. Such a simple proof, analogously as other proofs of same result known in literature, assumes that the well-known Shannon-McMillan-Breiman Theorem holds. The  name of this latter theorem  refers to the \emph{entropy 
theorem} for three different kinds of convergence.  Since we would like to get straight to the point of the matter, in this section we simply introduce 
the notation and main results that will be used in the following and then we will dedicate a whole section to a historical survey for this inequality that we consider one of the main contributions of this paper. We just point out here that Shields in \cite[Section II.$1$]{Shieldsbook96} gives two proofs of the Ornstein and Shields inequality: one of them assumes, as we do, that the entropy theorem holds; the second proof is quite long but it uses no previous deep result. Shields in his book also shows that the entropy theorem can be easily deduced from both the 1990 Ornstein and Shields inequality together with its complementary result that claims that universal codes reaching the entropy exist.

\medskip

For any notation not explicitly defined in this paper, we refer to \cite{Shieldsbook96}.

We recall the definition of \emph{typical set} and, since we will make use of it in the following, we state Shannon's Theorem 3 in \cite{Shannon}, therein proved by Shannon for independent and identically distributed (i.i.d. in short) sources, 
as being presented in  \cite{Cover:2006}. McMillan in \cite{mcmillan53} called it the Asymptotic Equiripartition Property, or AEP  in short. For any set $A$, $|A|$ denotes the cardinality of $A$ and $\mathcal{X}^n$ is the set of all the sequences of length $n$ of elements in $\mathcal{X}$.

Now we give the definition of typical set and we state Shannon's Theorem 3 for the case of i.i.d. sources.

\begin{definition}\label{def:tipical}
The \emph{typical set} $A^{(n)}_\epsilon$ with respect to $p(x)$ is the set of sequences \linebreak
$(x_1, x_2, \ldots , x_n)\in\mathcal{X}^n$ with the property 
$$
2^{-n(H(X)+\epsilon)}\leq p(x_1, x_2, \ldots,x_n)\leq 2^{-n(H(X)-\epsilon)}.
$$
\end{definition}
\begin{theorem}{\cite[Theorem 3]{Shannon}}\label{the:sha3}
\begin{enumerate}
\item[1)] If $(x_1,\ldots,x_n)\in A^{(n)}_\epsilon$, then 
$$
H(X) - \epsilon \leq -\frac{1}{n}\log_2 p(x_1,\ldots,x_n) \leq H(X) + \epsilon.
$$
\item[2)] $\mathrm{Pr}\left\{A^{(n)}_\epsilon\right\} > 1 - \epsilon$ for $n$ sufficiently large.
\item[3)] $\left | A^{(n)}_\epsilon \right | \leq 2^{n(H(X) + \epsilon)}$ for $n$ sufficiently large.
\item[4)] $\left | A^{(n)}_\epsilon \right | \geq (1-\epsilon) 2^{n(H(X) - \epsilon)}$ for $n$ sufficiently large.
\end{enumerate}
\end{theorem}

It is worth noticing that points 1) and 2) of the previous theorem represent a statement of~\cite[Theorem 3]{Shannon}, which will be exactly stated in the historical section, Section~\ref{sec:sur}, and points 3) and 4) are direct consequences of points 1) and 2).

Shannon's theorem is proved in~\cite{Shannon} for convergence in probability: in the body of the paper only for i.i.d. sources, and in Appendix 3 only for Markov-like sources that are ergodic and stationary. Later McMillan~\cite{mcmillan53} stated and proved  the \emph{entropy 
theorem} for stationary (not necessarily ergodic) processes with mean $L_1$-convergence. Finally, Breiman~\cite{breiman57,breiman60} proved the same result for stationary and ergodic processes and finite alphabets in the case of almost surely convergence; this latter result is extended by Chung \cite{Chung61,Chung62} to countably infinite alphabets. As mentioned above, these statements of the entropy 
theorems for different types of convergece are called \emph{Shannon-McMillan-Breiman Theorem} or the entropy Theorem as in \cite{Shieldsbook96}.

The following statement is the  entropy theorem in its almost surely version, or pointwise version, of Theorem~\ref{the:sha3} as being presented in \cite[Theorem I.7.1]{Shieldsbook96}.

\begin{theorem}\label{the:AEPalmost}
For any stationary and ergodic source and for each $\epsilon>0$,  $(x_1,\ldots,x_n)\in A^{(n)}_\epsilon$ eventually almost surely.
\end{theorem}

Clearly Theorem \ref{the:AEPalmost} implies Theorem \ref{the:sha3}.

\smallskip

 For sake of completeness anyhow we report from \cite{Shieldsbook96} the following definition and the Borel-Cantelli Lemma (cf. \cite[Lemma I.1.14]{Shieldsbook96}).

\begin{definition}
    A property $P$ is said to be \emph{measurable} if the set of all $\bf{x}$ for which $P(\bf{x})$ is true is a measurable set. If $\{P_n\}$ is a sequence of measurable properties then
    $P_n(\bf{x})$ holds \emph{eventually almost surely}, if for almost every $\bf{x}$ there is an $N=N(\bf{x})$ such that $P_n(\bf{x})$ is true for $n \geq N$.

\end{definition}

In Theorem \ref{the:AEPalmost} the property $P$ of ${\bf{x}}= (x_1,\ldots,x_n)$ is the membership in $A^{(n)}_{\epsilon}$.

\begin{lemma}[Borel-Cantelli]\label{le:borel-cantelli}
If $\{C_n \}$ is a sequence of measurable sets in a probability space $(X, \Sigma, \mu)$
such that $\sum \mu(C_n) < \infty$ then for almost every $x$ there is an $N=N(x)$ such that
$x \not\in {C_n}$ for $n \geq N$.
\end{lemma}

We now give the main definitions of this note. 

\smallskip

\noindent First, we recall that, given a finite set of symbols (or words) $\rchi$, the set $\rchi^{*}$ is the set of all sequences of of symbols (or words) of any length  $\rchi$, i.e. $\rchi^{*}= \bigcup_{n=0}^{+\infty} \rchi^{n}$.

\begin{definition}
	\label{compressor}
\begin{enumerate}
    \item[1)]  A binary \emph{faithful-code sequence} or \emph{one-to-one code sequence} is any function $\gamma$ from $\rchi^{*}$ to $\{0,1\}^{*}$ such that for any integer  $n$ its restriction to $\rchi^{n}$ is injective.
    \item[2)] A binary \emph{prefix-code sequence} is a one-to-one code sequence $\gamma$ such that for any integer  $n$ its restriction to $\rchi^{n}$ is a prefix code.
    \item[3)]	A binary \emph{lossless compressor} is any injective function $\gamma$ from $\rchi^{*}$ to $\{0,1\}^{*}$. 
\end{enumerate}

\end{definition}

Notice that  the notation  \emph{faithful-code sequence} in part 1) of Definition \ref{compressor} is used in \cite{OrnsteinShields90} and in \cite{Shieldsbook96} and  \emph{one-to-one code sequence}   is used  in \cite{blundo96} and in many other articles  (see \cite{kontoyiannis97,Kontoyiannis14}, references therein and citing articles). The same notion is also classically called  \emph{non-singular code sequence} (see for example \cite{Cover:2006} and again \cite{blundo96}).

For the sake of simplicity, we restrict our attention to binary codes. All the notations and results presented here can be extended to general finite alphabets, analogously as in Khinchin \cite{Khinchin1953}. 

\medskip

By definition the class of one-to-one code sequences strictly includes the class of prefix-code sequences \emph{and} the class of lossless compressors and, therefore, any result
that holds for the class of one-to-one code sequences holds also for the other two classes and it is, from a logical view point, a stronger result. Anyhow by adding some further reasoning and proofs it is sometimes possible to pass from a weaker to a stronger result. This has been done from an historical point of view and it is explained in this note in the following Subsection \ref{Subsec:Ksin}.

As the reader can notice, for us a lossless compressor is just an injective function, that grants for unique decodability. If the length of the message to encode is known to the decoder, also one-to-one code sequences grant for unique decodability.
Notice also that  a compression, in the usual meaning of this term, is not granted since a ``compressor", following above Definition \ref{compressor}, can even expand texts in average.

The paper is organized as follows. In Section~\ref{sec:OSin}, we derive from the Shannon-McMillan-Breiman Theorem an original proof of Ornstein-Shields inequality which we think to be simpler than the ones given in literature until now to our best knowledge.  Even if the inequality firstly stated by Khinchin~\cite{Khinchin1953}, called Khinchin inequality, follows from the Ornstein-Shields inequality, in Subsection~\ref{Subsec:average} we give, starting from Shannon's Theorem 3 in \cite{Shannon}, an original and elementary proof of such inequality 
 that avoids the use of measure theory tools such as the Lemma of Borel-Cantelli and 
that can be used as an exercise for Information Theory students.

Section~\ref{sec:sur} is entirely dedicated to a historical survey and some technical observations of such inequality. Once given an overview on the proofs of Ornstein-Shields inequality found in literature, we show in Subsection~\ref{sec:dif} why we think our proof is simpler than the others. 

The last two sections are dedicated to the aknowledgements and to some memories of the third author concerning Professor Aldo de Luca.

 \section{Simple proofs}\label{sec:OSin}
  
 The following theorem is due to Ornstein and Shields in \cite{OrnsteinShields90} and we give a simpler original proof. It holds for any stationary and ergodic source of entropy H. 
 
 \smallskip
 
 To avoid confusion in notations we denote by $|| \cdot ||$ the length of a sequence of  symbols (or words)  instead of the more common $|\cdot |$ that is here used to denote the cardinality of a set.

 \begin{theorem}\label{the:Ornstein} For any one-to-one code sequence $\gamma$ almost surely,
$$ \liminf_{n\in \mathbb{N}}{\frac{||\gamma((x_1,\cdots,x_n))||}{n}}\geq H.$$ 
 \end{theorem}

\begin{proof} 
For any $\epsilon>0$ and for any $n$, define the sets $C_{\epsilon}^{(n)}$  
as
$$ C^{(n)}_\epsilon = \{(x_1,\ldots,x_n)\in  A^{(n)}_\epsilon  : ||\gamma(x_1,\ldots,x_n)|| \leq \log_2 (|A_{\epsilon}^{(n)}|) -3\epsilon n-1\}.$$

%

Since $\gamma$ is one-to-one, $|C^{(n)}_\epsilon|= |\gamma(C^{(n)}_\epsilon)|$. Since $\gamma(C^{(n)}_\epsilon)$ is a subset of $\{ 0, 1\}^*$, for any $t$ it must contain less than $2^{t+1}$ strings of length at most $t$. 
Thus $|C^{(n)}_\epsilon|\leq |A_{\epsilon}^{(n)}|2^{-3\epsilon n}.$

Since $C^{(n)}_\epsilon\subseteq A^{(n)}_\epsilon$, each element in $C^{(n)}_\epsilon$ has a bound on its probability given in Definition \ref{def:tipical}. Using it and the bound on $|A_{\epsilon}^{(n)}|$ given in part $3)$ of Theorem \ref{the:sha3}, we have that $P(C^{(n)}_\epsilon)\leq 
2^{-\epsilon n}$ for $n$ sufficiently large. For each fixed $\epsilon>0$ we  apply the Borel-Cantelli Lemma  to the sequence $C^{(n)}_\epsilon$ and, by Theorem \ref{the:AEPalmost}, eventually almost surely $(x_1,\ldots, x_n)$ belongs to $ A^{(n)}_\epsilon$ and not to $C^{(n)}_\epsilon$, i.e.  $||\gamma(x_1,\ldots,x_n)|| \geq \log_2 (|A_{\epsilon}^{(n)}|) -3\epsilon n$.

Using the bound on $|A_{\epsilon}^{(n)}|$ given in part 4) of Theorem \ref{the:sha3}, for each $\epsilon$,  almost surely 
$$
\begin{aligned}[t]
\liminf_ {n\in \mathbb{N}}\frac{||\gamma(x_1,\ldots, x_n)||}{n}& \geq  \liminf_ {n\in \mathbb{N}}\frac{\log_2 (|A_{\epsilon}^{(n)}|) -3\epsilon n }{n}
&\geq 
H(X) -4\epsilon,
\end{aligned}
$$ 
where the value $\frac{\log_2 (1-\epsilon)}{n}$ disappears in the $\liminf$ because $\epsilon$ is fixed and $\log_2 (1-\epsilon)$ is a constant.
Since this holds for any $\epsilon$,  and in particular for the enumerable sequence $\epsilon_m=\frac{1}{m}, m=1,\ldots,+\infty$,  
a simple exercise of measure theory completes the proof. 
\end{proof} 

\subsection{The average case}\label{Subsec:average}
Here we  give an elementary proof of the Khinchin result, i.e. that entropy is a lower bound for the average compression, without using the  Borel-Cantelli Lemma (Sec.~\ref{sec:intro}, Lemma~\ref{le:borel-cantelli}). We derive it directly from the Shannon Theorem (Sec. \ref{sec:intro}, Theorem \ref{the:sha3}) that is present in all textbooks of Information Theory (cf. for instance \cite{Cover:2006,yeung:2008,korner:2011,MacKay:2002}).

Clearly any average result follows from the analogous pointwise result and, in particular, the result of this subsection follows from Theorem \ref{the:Ornstein}. Anyway we have decided to keep both proofs, since the proof of the average result uses only elementary mathematical notions and it could be used as an exercise for Information Theory students analogously as in the case of Shannon's \cite[Theorem 4]{Shannon} (cf. \cite[Chapter 3, Exercise 11]{Cover:2006}). The average result implies as corollaries classical Information Theory results such as, for instance, the fact that the entropy is a lower bound of the average length  of  uniquely decodable block codes or a lower bound for the compression ratio of arithmetic compressors.

\medskip

Notice that we make use the same idea of the proof of Theorem \ref{the:Ornstein}.
and, indeed, the first six lines of both proofs are exactly the same.

\begin{theorem}\label{the:main2} For all i.i.d. source of entropy $H$ and any one-to-one code sequence $\gamma$ 
$$\liminf_{n \in \mathbb{N}}\frac{1}{n} \sum_{\pmb x \in \rchi^{n}} ||\gamma(\pmb x)|| \cdot p(\pmb x)\geq H.$$
\end{theorem}

\begin{proof}

Proceed as in the proof of Theorem \ref{the:Ornstein} up to the fact that $P(C^{(n)}_\epsilon)\leq 2^{-\epsilon n}\leq \epsilon$ for $n$ sufficiently large. 

For any $n$, 
\begin{align*}
    \frac{1}{n} \sum_{\pmb x \in \rchi^{n}} ||\gamma(\pmb x)|| \cdot p(\pmb x)&\geq \frac{1}{n} \sum_{\pmb x \in A^{(n)}_\epsilon\setminus C^{(n)}_\epsilon} ||\gamma(\pmb x)|| \cdot p(\pmb x) \\
    &\geq \frac{\log_2 (|A_{\epsilon}^{(n)}|) -3\epsilon n}{n} \sum_{\pmb x \in A^{(n)}_\epsilon\setminus C^{(n)}_\epsilon}  p(\pmb x)\\
    &= [\frac{\log_2 (|A_{\epsilon}^{(n)}|)}{n} -3\epsilon][P(A_{\epsilon}^{(n)})- P(C_{\epsilon}^{(n)})].
\end{align*}

For $n$ sufficiently large,
 using part $4)$ of Theorem \ref{the:sha3}, one has  $\frac{\log_2 (|A_{\epsilon}^{(n)}|)}{n}\geq H-2\epsilon$. For $n$ sufficiently large,
 using part $2)$ of Theorem \ref{the:sha3} and previous obtained bound on $P(C_{\epsilon}^{(n)})$, one has  that for any $\epsilon < \frac{1}{2}$ 
 \begin{equation*}\label{equation:ave}
     \liminf_{n \in \mathbb{N}}\frac{1}{n} \sum_{\pmb x \in \rchi^{n}} ||\gamma(\pmb x)|| \cdot p(\pmb x)\geq (H-5\epsilon)(1-2\epsilon).
 \end{equation*}

Since this holds for any $\epsilon<\frac{1}{2}$, a simple exercise of calculus completes the proof.
\end{proof}

 From the previous proof we can derive, by Theorem \ref{the:sha3}, the desired average lower bound for compression. It is worth to highlight that the previous proof can be used to prove the same average lower bound also for stationary and ergodic sources, analogously as the original Khinchin proof, since Theorem~\ref{the:sha3} can be generalized to these kind of sources. Notice finally that the original Khinchin result was proved for stationary Markov chains.

\section{Historical and technical notes}\label{sec:sur}

\subsection{Historical survey}\label{subsec:sur}
Theorem \ref{the:Ornstein}  has been proved for the first time in \cite{OrnsteinShields90} in 1990 by Ornstein and Shields. More precisely it is part of their Theorem 1 in the invertible case, that is proved in \cite[Section 2]{OrnsteinShields90}.  Another simple proof is given in the  book of Shields \cite{Shieldsbook96} in 1996, together with a second  proof that is similar to the original proof given in
 \cite[Section 2]{OrnsteinShields90} and does not make use of the entropy theorem. The inequality  is stated in Theorem $II.1.2$ of \cite{Shieldsbook96} and the two proofs are given respectively in Subsection $II.1.b$ and in Subsection $II.1.c$. A third proof is given in \cite{Kontoyiannis14} in 2014 by Kontoyiannis  and Verd\`u in part $ii)$ of their Theorem $12$. Those three above are all the  proofs of Theorem \ref{the:Ornstein} that are present in the literature to our best knowledge. 

As noticed by Kontoyiannis and Verd\`u in \cite{Kontoyiannis14} before their Theorem 12, the weaker  corresponding result for prefix code sequences instead than for one-to-one sequences was established in \cite{kontoyiannis97}, \cite{Barronthesis} and  \cite{kieffer91}. More precisely in \cite{Barronthesis}, that is the 1985 Ph.D. thesis of Barron, it is proved a lemma that has as  an easy consequence the analogous of Theorem \ref{the:Ornstein} for prefix code sequences. Kontoyiannis uses it in his $1997$ paper \cite{kontoyiannis97}  and claims that: ``It is an unpublished result that appeared in \cite{Barronthesis}, and also, in a more general form in \cite{Algoetthesis}", and then he gives a proof of it in the appendix. Indeed, it seems a bit earlier, Barron's Lemma was stated, again as unpublished result of Barron, and proved in \cite{Shieldsbook96}. It is worth noticing that 
the Barron's Ph.D. thesis \cite{Barronthesis} and the Algoet's Ph.D. thesis \cite{Algoetthesis} are of the same year, both from Stanford University and with the same supervisor, that is, Professor Thomas M. Cover. 

We further notice here that the weaker corresponding result for lossless compressors was stated and proved for the first time by  Khinchin in  \cite{Khinchin1953} in 1953 (see also \cite{Khinchin1957}). It is a weaker result not just for  the fact that it is proved only for lossless compressors, but also because it is a result that is stated  ``in average"  instead of being pointwise. Notice that none of previous cited paper that we examined cites, in turn,  Khinchin. 
Also for this reason, we consider this historical section as one of the main contribution of this paper.

It is worth highlighting that, due to the analogy with the Shannon-McMillan-Breiman Theorem, we decided to call Theorem \ref{the:Ornstein} the ``Khinchin-Ornsten-Shields inequality" even if some credits have to be given to Imre Csisz\`ar as it will be explained in the following Subsection \ref{Subsec:Ksin}.

\bigskip

Khinchin's result comes from the golden age of the beginning of Information Theory, Indeed Khinchin relates it to one of the first theorems of Shannon's seminal paper \cite{Shannon}. Indeed the following theorem   is the seminal version of the Shannon-McMillan-Breiman Theorem, whilst Khinchin obtains his inequality starting from the subsequent Shannon's \cite[Theorem 4]{Shannon}. Both theorems are often present in textbooks such as in \cite[Chapter 3]{Cover:2006}.

\begin{theorem}{\cite[Theorem 3]{Shannon}}\label{th:shaoriginal}
Given any $\epsilon>0$ and $\delta>0$, we can find $N_0$ such that the sequences of any lenght $N\geq N_0$ fall into two classes:
\begin{enumerate}
    \item A set whose total probability is less than $\epsilon$
    \item The remainder, all of whose members have probabilities satisfying the inequality
    \[
    \bigl|\frac{\log p^{-1}}{N}- H\bigr| <\delta.
    \]
\end{enumerate}
\end{theorem}
\noindent Notice that Theorem~\ref{the:sha3} is the same statement of the previous theorem considering $\delta=\epsilon$. 

Theorem~\ref{th:shaoriginal} is proved in~\cite{Shannon} for convergence in probability: in the body of the paper only for i.i.d. sources, and in Appendix 3 only for Markov-like sources that are ergodic and stationary. Later McMillan~\cite{mcmillan53}  showed for stationary (not necessarily ergodic) processes with mean $L_1$-convergence a simple corollary to the entropy ergodic theorem relating codebook sizes to the appropriate entropies, using measure theory and Martingale theory. Finally, Breiman~\cite{breiman57,breiman60} proved the same result for stationary and ergodic processes and finite alphabets in the case of almost everywhere convergence,  using again the Martingale theory. The Breiman's result is extended by Chung \cite{Chung61,Chung62} to countably infinite alphabets.

As already mentioned, these three statements of the entropy ergodic theorems for different kynds of convergece are called Shannon-McMillan-Breiman Theorem.

 \bigskip

 There is a wide literature on developing simple proofs for the entropy ergodic theorems and their generalizations. Usually these proofs are not short but often they use fairly elementary mathematics. Just to give an example of a simple proof of a generalization of the Shannon-McMillan theorem, we may cite a result proved by Kieffer in~\cite{kieffer1974}, where it is provided a short proof of the most general then-known related result which did not involve Martingale theory using an extension of a proof presented in~\cite[Theorem 3.5.3]{Gallager1968}. These results successfully tackled  the overall problem for the $L_1$ convergence case.

However, it is far beyond the scope of this section to give an historical survey of the Shannon-McMillan-Breiman Theorem and of its generalizations, extensions and consequences, but the interested reader can see also \cite{AlgoetCoverSandwich,Barronarticle,QuantumMcMillan} and citing articles.

In general, it is far beyond the scope of this paper to give a simple proof of the entropy theorem even if, as noticed by Shields in \cite{Shieldsbook96} and reported by us in  Section \ref{sec:intro}, the Ornstein-Shields inequality is deeply connected to the entropy theorem. We want to give a new simple proof of the Ornstein-Shields inequality that we believe to be simpler of the other existing proofs of the same kind which make use of the entropy theorem, as we will discuss in next Subsection \ref{sec:dif}.

\bigskip

Coming back to Shannon's fourth theorem, consider
for any $n$ the sequences of length $n$ to be arranged in order of  decreasing probability. For any $q$, $0<q<1$, define $n(q)$ to be the number we must take from this ordered set, starting with the most probable one, in order to accumulate a total probability $q$ for those taken. \cite[Theorem 4]{Shannon} states that $\lim_{n\to\infty} \frac{\log_2 n(q)}{n}=H.$

Shannon, after stating his Theorem 4, claims, without proving anything, that: ``We may interpret $\log_2 (n(q))$ as the number of bits required to specify the sequence when we consider only the most probable sequences with a total probability $q$. Then $\frac{\log_2 n(q)}{n}=H$ is the number of bits per symbol for the specification. The theorem say that for large $n$ this will be independent of $q$ and equal to $H$".

This Shannon's interpretation is correct and starting from this seminal claim several formal consequences have been proved  that usually concern compressors that code blocks of uniform length $n$ into blocks of uniform length $k$ (see for instance \cite{korner:2011,sgarro79}).


Instead, following above Shannon's interpretation as a research direction and indeed exploiting \cite[Theorem 4]{Shannon}, Khinchin in  \cite{Khinchin1953} 
gives a formal definition of the \emph{average compression} for sequences of fixed length and of the \emph{compression coefficient} as $\limsup$ of the average compression. Then he proves in \cite[Theorem 4]{Khinchin1953}, for the first time to our best knowledge, that the entropy is a lower bound for the  compression coefficient of any injective function. Khinchin's proof works also when  $\liminf$ is used in the place of $\limsup$ in the definition of average compression, as we did in Theorem \ref{the:main2} . A $\liminf$  gives indeed a stronger result from a logical point of view.

This  Khinchin's inequality is also proved in another context, i.e. in the case of linguistic sources, by Hansel-Perrin-Simon in \cite{Hansel1992CompressionAE}.

\subsection{From weaker to stronger inequality}\label{Subsec:Ksin}

In Section $2$ of the original $1990$ paper of Ornstein and Shields \cite{OrnsteinShields90} where it is given the first written statement and proof of Theorem \ref{the:Ornstein}, it is described a technique of transforming one-to-one code sequences into prefix-code sequences adding a ``small" overhead header. The description of this technique takes a good part of their Section $2$ and it is also described in the $1996$ book of Shields \cite{Shieldsbook96} in Subsection $I.7.d$; the ``small" overhead header is $O(\log(||\gamma(w)||))$   size for any $w\in\rchi^{*}$ and gives ``no change in asymptotic performance" as said in the book of Shields. 

There is no room in this short subsection to describe in details this technique that mainly consists in prepending to $\gamma(w)$ the Elias delta coding of $||w||$ \cite{elias75}. Here we want to notice that in \cite{OrnsteinShields90}
this technique is credited, in their Section $2$, to  Imre Csisz\`ar and, moreover, the authors  say in the acknowledgements: 
``We wish to give special thanks to Imre Csisz\`ar, who corrected several of our errors and made many suggestions for improvement of our discussion".

In the $1996$ book of Shields \cite{Shieldsbook96} the simpler proof among the two proof contained therein consists exactly in linking the technique of Csisz\`ar and the Barron's $1985$ result \cite{Barronthesis} and it will be discussed in Subsection \ref{sec:dif}.

\subsection{Analogies and differences between proofs}\label{sec:dif}

What is a ``simple" proof? Can we say that a simple and elegant proof that is given at the end of a mathematical book and that makes use of all previous results of that book is  really ``simple"?

We think that the second question has a ``no" as right answer  and we have no answer to first question. Maybe a good attempt is given by the Occam's priciple or ``razor" discussed also in Barron's Ph.D. thesis  \cite{Barronthesis}  that can give suggestions to decide when a proof is ``simpler" than another. In this subsection we analyze the three proofs that are known before our proof and compare them all. 

\bigskip

As we said above, in~\cite[Section II.1.b, Section II.1.c]{Shieldsbook96} are reported two proofs of Theorem~\ref{the:Ornstein}. One short and elegant and the second longer 
 ``which was
developed in \cite{OrnsteinShields90}, and does not make use of the Shannon-McMillan-Breiman Theorem $[\ldots]$" as Shields wrote.

Let us firstly discuss the second longer proof. It is more than four pages long in the book of Shields and also refers to some previous notations and results but it does not uses previous deep results. Even if it does not make use of the Shannon-McMillan-Breiman Theorem we cannot consider it simpler than our proof. We think that it is not possible to claim that one of these two proofs is simpler than the other.  We emphatize that the beauty of this long proof resides also in the fact that just after it, in the book of Shields, it is given a short proof that logically derives the Shannon-McMillan-Breiman Theorem from Theorem \ref{the:Ornstein} and from the converse inequality stated in Theorem $II.1.1$ of \cite{Shieldsbook96}. This reasoning shows the generality and the logical   power of Theorem \ref{the:Ornstein}.

 \smallskip
 
 Let us now analyze the three remaining proofs: the first elegant proof written in the book of Shieds in $1996$, the $2014$  Kontoyiannis and Verd\`u's proof and our proof of Theorem \ref{the:Ornstein}. All three use the Borel-Cantelli Lemma (Lemma~\ref{le:borel-cantelli}) and the Shannon-McMillan-Breiman Theorem.
 
 Concerning the first elegant proof in \cite{Shieldsbook96}, 
 Shields claims that Theorem \ref{the:Ornstein} ``will follow from the entropy theorem, together
with a surprisingly simple lower bound on prefix-code word length."
 This proof indeed 
 makes use of several components:
 \begin{itemize}
     \item[$1)$] For any $n$, there is a conversion of one-to-one code into a prefix code ``with no change in asymptotic performance"  
     by using the 
Imre Csisz\`ar  technique (cf. Subsection \ref{Subsec:Ksin}).

\item[$2)$] The use of a Barron's lemma proved in \cite{Barronthesis} (see also \cite[Lemma II.1.3]{Shieldsbook96}) that, in turns uses \begin{itemize}
    \item[$2.a)$]the  Kraft inequality for prefix codes \cite{kraft1949device};
    \item[$2.b)$] The Borel-Cantelli Lemma (Lemma~\ref{le:borel-cantelli}).
\end{itemize}
\item[$3)$] The Shannon-McMillan-Breiman Theorem (Theorem~\ref{the:AEPalmost}).

 \end{itemize}

We want to emphasize here that our simple proof of Theorem \ref{the:Ornstein} makes use only of above points $2.b)$ and $3)$ analogously as the Kontoyiannis and Verd\`u's proof and it is overall shorter even if the length of the proof of the Kraft inequality is not considered. Therefore we think that our proof is simpler than the first of the two proofs contained in the book of Shields. A last argument in favour of our thinking is explained in the last part of this subsection.

\smallskip 

Let us now examine the proof of Kontoyiannis and Verd\`u of part $ii)$ of \cite[Theorem $12$]{Kontoyiannis14}, that is exactly our Theorem \ref{the:Ornstein}. Their proof makes use in turns of \cite[Theorem $11$]{Kontoyiannis14} that, again in turns, uses \cite[Theorem $5$]{Kontoyiannis14} that ``is a natural analog of the corresponding converse established for prefix compressors in \cite{Barronthesis}'' by Barron, as authors say. Their elegant proof of \cite[Theorem $5$]{Kontoyiannis14} uses a counting argument  to generalize the Barron's lemma and avoids the use of the Kraft inequality.
The resulting proof of part $ii)$ of \cite[Theorem $12$]{Kontoyiannis14}, even including  all these backpointers to previous theorems and their proofs, turns out to be overall just a bit longer than  the elegant proof reported in the Shield's book \cite{Shieldsbook96} but at least it does not make use of Kraft inequality.

\medskip

As final argument we notice that 
all above two proofs~\cite[Section II.1.b]{Shieldsbook96} and~\cite[Theorem 12]{Kontoyiannis14}, but not our proof, 
in the first part  show that the probability of the set of sequences of 
length $n$ that  have a ``small" compressed length is summable in $n$, and this  allows the use of Borel-Cantelli Lemma. Then such proofs use the Shannon-McMillan-Breiman to obtain the result. Somehow this procedure is analog to what Khinchin  does in the average case by exploiting \cite[Theorem 4]{Shannon}.

Our proof  in the first part shows instead  that the probability of the set of sequences of length $n$ that \emph{1)} have a small compressed length, \emph{and 2)} are typical, is summable  in $n$ and this fact allows us to use the bounds on the probability of each element in the typical set and, consequently, to simplify the proof.

Clearly we prove a weaker result in a simpler way but this weaker result still allows us to obtain the desired pointwise inequality. 
We think that maybe the Shannon's interpretation reported at the end of  Subsection \ref{subsec:sur} directed the other two proofs along the lines of a stroger result, lines that were also followed by Khinchin in the proof of his average result.

\section{Conclusion}
In this paper we derived from the entropy theorem  for i.i.d. sources in the case of convergence in probability, i.e. the Shannon's Theorem 3 in \cite{Shannon}, and from the more general entropy theorem for stationary and ergodic processes in the case of almost surely convergence, also called McMillan-Breiman Theorem \cite{mcmillan53,breiman57,breiman60}, respectively an elementary proof for the Khinchin inequality and a simple proof  for the Ornstein-Shields inequality. In particular, we use a deeper classical result in order to prove the latter inequality which is the almost-sure version of the first one, that is an average inequality.

 \section{In memoriam of Professor Aldo de Luca}

The third author remembers the discussions and explanations of Aldo
de Luca given to him around thirty years ago while walking in Boulevard Saint Michel in Paris. Aldo was very fond of Khinchin's formalization effort and, indeed, in his research paper \cite{deLuca} he uses the Khinchin's notation ``standard sequences" reported in the English translation of Khinchin's work, instead of the more common ``typical sequences". 

Aldo's voice 
had a seducing sound, similar to the sound of a father reading a beautiful fairy tale to his sons, or an history of brave knights fighting for honour and mathematical rigour.
 
 \section{Acknowledgements} The authors thank Professor D. Perrin for pointing out reference \cite{Shieldsbook96} during a conference in Rome, July 11-12 2019, in memoriam of Professor de Luca, where they presented a preliminary version of above results. The authors are also grateful to the referees for their suggestions.

\nocite{*}
\bibliography{bibliografia} 
\bibliographystyle{ieeetr}
 
\end{document}